\documentclass{article}
\usepackage[preprint, nonatbib]{neurips_2024}

\usepackage[makeroom]{cancel}
\usepackage[numbers]{natbib}
\usepackage{natbib}
\usepackage{hyperref}
\usepackage{amsmath}
\usepackage{amssymb}
\usepackage{enumitem}
\usepackage{amsthm}
\usepackage{tikz}
\usepackage{algorithm}
\usepackage{algpseudocode}
\usepackage{caption}
\usepackage{subcaption}
\usepackage{booktabs}
\usepackage{pgfplots}
\usepackage{cleveref}
\usepgfplotslibrary{external}
\pgfplotsset{compat=1.17}
\usepackage{pgfplotstable}
\usepgfplotslibrary{groupplots}
\usepgfplotslibrary{external}
\newtheorem{theorem}{Theorem}[section]
\newtheorem{proposition}[theorem]{Proposition}
\newtheorem{lemma}[theorem]{Lemma}
\newtheorem{corollary}[theorem]{Corollary}
\newtheorem{definition}[theorem]{Definition}

\crefname{theorem}{Theorem}{Theorems}
\Crefname{theorem}{Theorem}{Theorems}
\crefname{proposition}{Proposition}{Propositions}
\Crefname{proposition}{Proposition}{Propositions}
\crefname{lemma}{Lemma}{Lemmas}
\Crefname{lemma}{Lemma}{Lemmas}
\crefname{corollary}{Corollary}{Corollaries}
\Crefname{corollary}{Corollary}{Corollaries}
\crefname{definition}{Definition}{Definitions}
\Crefname{definition}{Definition}{Definitions}
\crefname{figure}{Figure}{Figures}
\Crefname{figure}{Figure}{Figures}
\crefname{table}{Table}{Tables}
\Crefname{table}{Table}{Tables}
\crefname{section}{Section}{Sections}
\Crefname{section}{Section}{Sections}
\crefname{algorithm}{Algorithm}{Algorithms} 
\Crefname{algorithm}{Algorithm}{Algorithms} 

\DeclareMathOperator*{\argmax}{argmax}
\DeclareMathOperator*{\argmin}{argmin}
\algnewcommand\algorithmicforeach{\textbf{for each}}
\algdef{S}[FOR]{ForEach}[1]{\algorithmicforeach\ #1\ \algorithmicdo}

\title{Efficient algorithms for the sensitivities of the Pearson correlation coefficient and its statistical significance to online data}

\author{
    Marc Harary \\
	Department of Data Science \\
	Dana-Farber Cancer Institute \\
	Boston, MA 02115 \\
	\texttt{marc@ds.dfci.harvard.edu}
}

\begin{document}

\maketitle

\begin{abstract}
	Reliably measuring the collinearity of bivariate data is crucial in statistics, particularly for time-series analysis or ongoing studies in which incoming observations can significantly impact current collinearity estimates. Leveraging identities from Welford's online algorithm for sample variance, we develop a rigorous theoretical framework for analyzing the maximal change to the Pearson correlation coefficient and its p-value that can be induced by additional data. Further, we show that the resulting optimization problems yield elegant closed-form solutions that can be accurately computed by linear- and constant-time algorithms. Our work not only creates new theoretical avenues for robust correlation measures, but also has broad practical implications for disciplines that span econometrics, operations research, clinical trials, climatology, differential privacy, and bioinformatics. Software implementations of our algorithms in Cython-wrapped C are made available at \href{https://github.com/marc-harary/sensitivity}{\texttt{https://github.com/marc-harary/sensitivity}} for reproducibility, practical deployment, and future theoretical development.
\end{abstract}

\section{Introduction}
\label{sec:intro}

The Pearson correlation coefficient (PCC), one of the simplest measures of collinearity between two random variables, has been an essential component of data analysis for over a century \cite{pearson1895vii}. Along with other collinearity measures, it finds applications in quantitative disciplines that span physics, psychology, economics, and mechanical engineering \cite{casella2024statistical}. Like many statistical quantities, however, its utility can be significantly curtailed by the presence of outliers---data points that differ widely from the underlying distribution from which the majority of a dataset is sampled \cite{huber2004robust}. This issue is exacerbated by online data collection, during which new---and potentially outlying---data points have the potential to significantly alter the overall sample PCC.

Such scenarios are common across fields that rely on time-series data or otherwise ongoing observations, necessitating an understanding of the extent to which online PCC estimates are susceptible to change. For example, in medicine, linear relationships between patient covariates and clinical outcomes are often continuously monitored during drug development trials \cite{sliwoski2014computational, drews2000drug, hughes2011principles}. In finance and econometrics, investors dynamically manage their portfolios based in part on the PCC between assets, but are also forced to adjust these estimates given changing market conditions, evolving relationships between financial instruments, and ongoing data collection \cite{gujarati2009basic, hayashi2011econometrics, karpoff1987relation}.

 Obviously, frequentist statistics provides many common methods to quantify the certainty of sample estimates, including standard confidence intervals and hypothesis-testing \cite{casella2024statistical}. As for the problem of accounting for outliers, robust statistics \cite{huber2004robust, maronna2019robust} offers both simple and sophisticated analytical methods for estimating parameters that might be affected by outlying data. For instance, winsorization, analogous to clipping in signal processing, offers a straightforward technique for minimizing these effects \cite{dixon1974trimming}. Order statistics \cite{pitas1992order}, and their generalizations in the form of M-statistics and other measures \cite{huber2004robust, maronna2019robust}, offer a high breakdown-point, which quantifies the fraction of data points that are required to be outliers before the metrics' values extend beyond the range of inlying data \cite{huber2004robust, maronna2019robust}. Similar measures like Cook's distance \cite{cook2000detection} and empirical impact functions \cite{greene2003econometric} provide alternative means of gaining insight into the effect of outliers. For the specific case of linear relationships, \cite{rousseeuw2005robust} developed efficient trimming methods like Least Trimmed Squares (LTS) for simultaneously identifying and rejecting outlying points in regression analysis.

However, none of these approaches either directly addresses changes in sample estimates as a result of dynamic datasets \textit{per se} or offers an elegant method for computing the worst-case change across all possible observations. Measures of statistical significance like p-values also evolve as a function of incoming data, meaning that if the desired metric is not the PCC itself but rather its statistical significance, hypothesis-testing does not address the challenge of learning the target metric online. While robust statistics can provide invaluable insight into the effect of outliers on datasets, it assumes that all data points are fixed and \textit{a fortoriori} that outliers are already present in the sample \cite{huber2004robust, maronna2019robust}. Moreover, many robust estimators can be quite intensive to compute \cite{huber2004robust, maronna2019robust}, and unlike their far more popular and common non-robust counterparts like the PCC, are less immediately interpretable. As is the case for robust statistics, influence measures like empirical impact functions are also intended for scenarios in which known outliers already contaminate the data rather than for dynamic settings in which statisticians must simultaneously consider the entire space of incoming points.

In the following work, we derive analytical solutions to the sensitivity of the PCC and p-value of datasets with respect to the addition of new data, including for the worst-case observations that will induce the maximal changes. By drawing from Welford's online algorithm for variance \cite{welford1962note}, we derive a rigorous framework for analyzing the sensitivity of online collinearity measures. Our main results are that the solutions for the worst-case points are not only closed-form, but easily computable in either linear or constant time with a high degree of numerical precision.

This paper is structured as follows. In \cref{sec:prereq}, we establish prerequisite definitions, clarify notation, and state preliminary theorems. Our primary results are in \cref{sec:res}, in which we develop the closed-form solutions to primary sensitivity. \cref{sec:algo} provides pseudocode, a proof of correctness, and runtime analysis for an algorithm to efficiently compute these solutions. Experiments are conducted on synthetic and real-world data in \cref{sec:exp}, which are then followed by a brief discussion and suggestions for future directions in \cref{sec:disc}. We conclude in \cref{sec:conc}.

\section{Prerequisites}
\label{sec:prereq}

\begin{table}[ht]
\centering
\caption{Key notation}
\label{tab:notation}
\begin{tabular}{@{}cl@{}}
\toprule
Symbol & Description \\ \midrule
\multicolumn{2}{c}{\textbf{Topology of Feasible Region}} \\
$\mathcal{F}$ & Feasible region \\
$\text{int} \left( \mathcal F \right)$ & Interior of $\mathcal F$ \\
$\partial \mathcal F$ & Boundary of $\mathcal F$ \\
$\partial \mathcal U_y, \partial \mathcal L_y, \partial \mathcal U_x, \partial \mathcal L_x$ & Upper, lower, right, left edges of $\partial \mathcal F$ \\
$\ell_x, u_x, \ell_y, u_y$ & Lower, upper bounds for $x$, $y$ \\
$\mathcal C = \left\{ C_{\ell, \ell}, C_{\ell, u}, C_{u, \ell}, C_{u, u} \right\}$ & Corner points \\
$\mathcal I_{n-1} = \left\{ I_{x, \ell, n-1}, I_{x, u, n-1}, I_{y, \ell, n-1}, I_{y, u, n-1} \right\}$ & Intersection points \\ \midrule
\multicolumn{2}{c}{\textbf{Welford's Quantities}} \\
$s_{x,n-1}, s_{y,n-1}$ & Biased sample standard deviations of $D$\\
$s_{xy,n-1}$ & Biased sample covariance of $D$ \\
$\bar{x}_{n-1}, \bar{y}_{n-1}$ & Sample means of $D$ \\ \midrule
\multicolumn{2}{c}{\textbf{Statistical Measures}} \\
$D$ & Bounded multiset \\
$n-1$ & Cardinality of $D$ \\
$\left(x_n, y_n \right)$ & New data point to add to $D$ \\
$r_{n-1} \equiv r(D)$ & PCC of $D$ \\
$p_{n-1} \equiv p(D)$ & p-value of $r_{n-1}$ \\
$r_{n} \left( x_n, y_n \right) \equiv r \left( D \cup \left\{ (x_n, y_n) \right\} \right)$ & PCC of $D$ with new point \\
$p_{n} \equiv p\left(D \cup \left\{ (x_n, y_n) \right\} \right)$ & p-value of $r_{n}$ \\
$\Delta_1 r(D, \mathcal{F}), \ \Delta_1 p(D, \mathcal{F})$ & Primary sensitivities $D$ within $\mathcal{F}$ \\
$\Delta_k r(D, \mathcal{F}), \ \Delta_k p(D, \mathcal{F})$ & $k$-ary sensitivities $D$ within $\mathcal{F}$ \\ \midrule
\multicolumn{2}{c}{\textbf{Regression Parameters}} \\
$\hat{\alpha}_{x, n-1}, \hat{\beta}_{x, n-1}, \hat{\alpha}_{y, n-1}, \hat{\beta}_{y, n-1}$ & BLUE values regressing $y$ onto $x$, $x$ onto $y$ \\
$\mathcal G_{x, n-1}, \mathcal G_{y, n-1}$ & Subset of $\mathcal F$ lying on BLUE lines \\ \midrule
\multicolumn{2}{c}{\textbf{Miscellaneous}} \\
$\mathbf H_{r_n} \left( x_n, y_n \right)$ & Hessian matrix of $r_n$ at $\left( x_n, y_n \right)$ \\
$\cup$ & Union operator (for both sets and multisets) \\
$| \cdot |$ & Cardinality (for both sets and multisets) \\
\bottomrule
\end{tabular}
\label{tab:not}
\end{table}

Concretely, we consider a bounded dataset (\cref{tab:not})
\begin{equation}
	D = \{(x_i, y_i)\}_{i=1}^{n-1} \quad \text{where} \quad \ell_x \leq x_i \leq u_x \quad \text{and} \quad \ell_y \leq y_i \leq u_y \quad \forall i
\end{equation}
and put the feasible region (\cref{fig:2d})
\begin{equation}
	\mathcal F := \left[ \ell_x, u_x \right] \times \left[ \ell_y, u_y \right].
\end{equation}
Note that in this section, for brevity, it is easier to define $n$ such that $|D| = n - 1$. We also note that $D$ is not a set in the traditional sense, but rather a multiset \cite{blizard1989multiset}; moreover, we abuse notation slightly by indexing each point $(x_i, y_i)$ and by using superscripts to indicate the Cartesian product of $\mathcal F$ with itself in the context of multisets. We define the sample correlation coefficient of $D$ as follows.

\begin{definition}
\label{def:pcc}
The Pearson correlation coefficient (PCC) of a dataset $D$ is given by
\begin{equation}
	r \left( D \right) := \frac{\sum_{i=1}^{n-1} \left( x_i - \bar x_{n-1} \right) \left( y_i - \bar y_{n-1} \right)}{\sqrt{\sum_{i=1}^{n-1} \left( x_i - \bar x_{n-1} \right)^2 \sum_{i=1}^{n-1} \left( y_i - \bar y_{n-1} \right)^2}},
\end{equation}
where $\bar x_{n-1}$ and $\bar y_{n-1}$ are the sample means of $D$ \cite{pearson1895vii}.
\end{definition}
This definition makes our motivation for imposing constraints on our datasets more obvious. Clearly, if incoming observations were not bounded, worst-case points could be infinitely large, overshadowing the present data and skewing the PCC to an arbitrary degree towards -1 or 1. We also opt for a rectangular feasible region due to its simplicity and similarity to winsorization \cite{dixon1974trimming}. For our proofs below, assuming this rectangular region corresponds to the unit square also facilitates our invocation of key theorems in real analysis \cite{smart1980fixed}. We nevertheless consider more complex boundaries in \cref{sec:disc}.

We also enumerate Student's t-test for ease of reference below. The monotonic relationship of the p-value with the PCC will also dramatically simplify our analyses.
\begin{theorem}[Student's t-test for the Pearson correlation coefficient]
\label{thm:ttest}
The significance of a PCC $r_{n-1}$ is given by
\begin{equation}
	p(D) \equiv p_{n-1} = 2 \left(1 - F_t(|t_{n-1}|, \text{df}_{n-1})\right),
\end{equation}
where 
\begin{equation}
t_{n-1} = r_{n-1} \sqrt{\frac{n - 3}{1 - r_{n-1}^2}},
\end{equation}
\begin{equation}
	\text{df}_{n-1} = n - 3,
\end{equation}
and $F_t$ is the cumulative density function (CDF) of Student's t-distribution \cite{casella2024statistical}.
\end{theorem}

\begin{corollary}
\label{cor:mono}
$p_{n-1}$ decreases monotonically with $|r_{n-1}|$.
\end{corollary}

In online settings, we consider the effect of adding a new point $\left( x_n, y_n \right)$ to $D$ on the PCC and its p-value, with an interest in the robustness of the PCC to the worst case. We formalize this as follows.
\begin{definition}
\label{def:primary}
	We call
	\begin{equation}
		\Delta_1 r \left( D, \mathcal F \right) := \max_{\left( x_n, y_n \right) \in \mathcal F} \lvert r \left(  D \right) - r \left(  D \cup \left\{ \left( x_n, y_n \right) \right\} \right) \rvert
	\end{equation}
	and
	\begin{equation}
		\Delta_1 p \left( D, \mathcal F \right) := \max_{\left( x_n, y_n \right) \in \mathcal F} \lvert p \left(  D \right) - p \left(  D \cup \left\{ \left( x_n, y_n \right) \right\} \right) \rvert
	\end{equation}
	the primary sensitivities of $D$ within $\mathcal F$.
\end{definition}

For the more general case of adding up to $k$ points to $D$, we define the $k$-ary sensitivity, to which we will also return briefly in \cref{sec:disc}.
\begin{definition}
\label{def:kary}
	We call
	\begin{equation}
		\Delta_k r \left( D, \mathcal F \right) := \max_{D^\prime \in \mathcal F^k} \lvert r \left(  D \right) - r \left(  D \cup D^\prime \right) \rvert
	\end{equation}
	and
	\begin{equation}
		\Delta_k p \left( D, \mathcal F \right) := \max_{D^\prime \in \mathcal F^k} \lvert p \left(  D \right) - p \left(  D \cup D^\prime \right) \rvert
	\end{equation}
	the $k$-ary sensitivities of $D$ within $\mathcal F$.
\end{definition}

Drawing from Welford's online algorithm for variance and covariance, we leverage the following recurrence relations between $D$ and subsequent data points \cite{welford1962note}.
\begin{theorem}[Welford's recurrence relations]
\label{thm:welford}
	Given a dataset $D = \{(x_i, y_i)\}_{i=1}^{n-1}$ and new point $\left( x_n, y_n \right)$, let $s_{x, n-1}$ and $\bar x_{n-1}$ be the biased sample standard deviation and sample mean, respectively, of $\left\{ x_1, x_2, \ldots, x_{n-1} \right\}$ and $s_{xy, n-1}$ be the biased sample covariance of $D$. Let $s_{x, n}$ and $\bar x_n$ be the biased sample standard deviation and sample mean, respectively, of $\left\{ x_1, x_2, \ldots, x_{n-1}, x_n \right\}$ and let $s_{xy, n}$ be the biased sample covariance of $D \cup \left\{ \left(x_n, y_n \right) \right\}$. Then
\begin{equation}
	\bar x_n = \frac{(n-1) \bar x_{n-1} + x_n } {n},
\end{equation}
\begin{equation}
	s^2_{x, n} = \frac{ (n-1) s^2_{x, n-1} + (x_{n+1} - \bar x_{n-1})(x_{n+1} - \bar x_n) }{n},
\end{equation}
and
\begin{equation}
	s_{xy, n} = \frac{n-1}{n} \left( s_{xy, n-1} + \frac{(x_n - \bar x_{n-1}) (y_n - \bar y_{n-1})}{n} \right).
\end{equation}
\end{theorem}

We will also find that the best linear unbiased estimators (BLUE) for the regression line for a dataset $D$ play a role in proving our key theorems \cite{casella2024statistical}. The BLUE parameters can be related to the statistical moments computed by Welford via the Gauss-Markov theorem \cite{casella2024statistical}.
\begin{theorem}[Gauss-Markov theorem]
\label{thm:gauss}
Fix a dataset $D = \{(x_i, y_i)\}_{i=1}^{n-1}$. Given a model function
\begin{equation}
	y_i = \alpha_{x, n-1} + \beta_{x, n-1} x_i
\end{equation}
regressing $\left\{ y_1, y_2, \ldots, y_{n-1} \right\}$ onto $\left\{ x_1, x_2, \ldots, x_{n-1} \right\}$, the least-squares line
is given by
\begin{equation}
	\hat \beta_{x, n-1} = \frac{s_{xy,n-1}}{s_{x,n-1}^2} \quad \text{ and } \quad \hat \alpha_{x, n-1} = \bar y_{n-1} - \beta_{x, n-1} \bar x_{n-1}.
\end{equation}
\end{theorem}

As is often noted in discussions of Gauss-Markov, we observe the following fact, which we formalize as a proposition for ease of reference in \cref{sec:res}.
\begin{proposition}
\label{prop:blueinter}
The sets $\mathcal G_{x,n-1}, \mathcal G_{y, n-1} \subset \mathcal F$ lying on the BLUE lines intersect at the point $(\bar x_{n-1}, \bar y_{n-1})$.
\end{proposition}
Likewise, we formally state another trivial property of linear regression.
\begin{proposition}
\label{prop:bluesame}
$\mathcal G_{x,n-1} = \mathcal G_{y, n-1}$ if and only if $r \left( D \right) = \pm 1$.
\end{proposition}

As to primary sensitivity itself, the following points will play a critical role.
\begin{definition}
\label{def:corner}
	We refer to the following set as corner points:
	\begin{equation}
		\mathcal{C} := \left\{
		\begin{aligned}
			C_{\ell, \ell} &:= \left( \ell_x, \ell_y \right), \\
			C_{\ell, u} &:= \left( \ell_x, u_y \right), \\
			C_{u, \ell} &:= \left( u_x, \ell_y \right), \\
			C_{u, u} &:= \left( u_x, u_y \right)
		\end{aligned}
		\right\}.
	\end{equation}
\end{definition}

\begin{definition}
\label{def:inter}
We refer to the following set $\mathcal{I}_n$\footnote{Note that $\mathcal{I}_n$ will vary depending on the value of $D$ and therefore evolve accordingly. Hence, we include the subscript $n$ to specify the dataset to which it belongs.} as intersection points:
\begin{equation}
\mathcal{I}_n := \left\{
\begin{aligned}
    I_{x, \ell, n-1} &:= \left( \hat{\beta}_{x, n}^{-1} \left( \ell_y - \hat{\alpha}_{x, n} \right), \ell_y \right), \\
    I_{x, u, n-1} &:= \left( \hat{\beta}_{x, n}^{-1} \left( u_y - \hat{\alpha}_{x, n} \right), u_y \right), \\
    I_{y, \ell, n-1} &:= \left( \ell_x, \hat{\beta}_{y, n}^{-1} \left( \ell_y - \hat{\alpha}_{y, n} \right) \right), \\
    I_{y, u, n-1} &:= \left( u_x, \hat{\beta}_{y, n}^{-1} \left( u_x - \hat{\alpha}_{y, n} \right) \right)
\end{aligned}
\right\}.
\end{equation}
\end{definition}

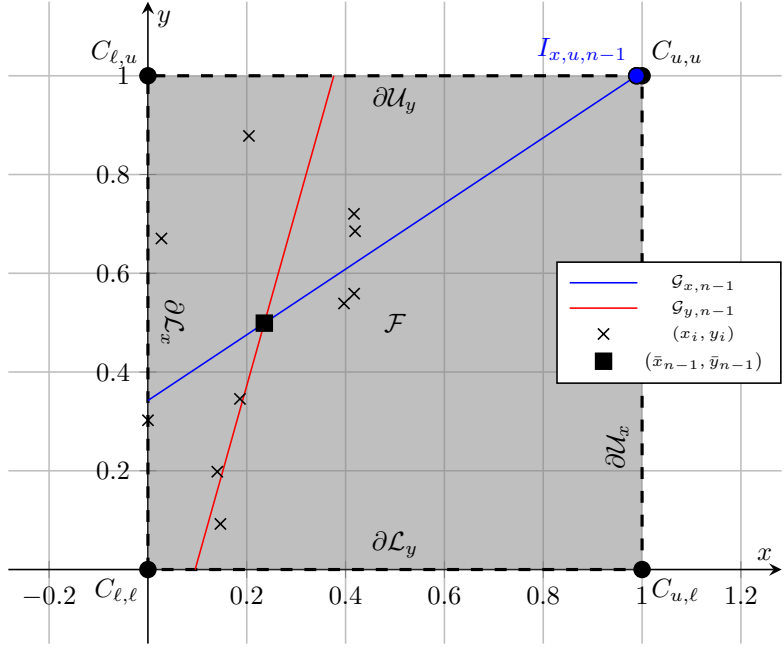
\begin{figure}[t]
    \centering
		\begin{tikzpicture}[scale=1.5, every node/.append style={scale=0.66667}]
			\pgfmathsetmacro{\n}{10}
			\pgfmathsetmacro{\xprevbar}{0.23556460542404717}
			\pgfmathsetmacro{\yprevbar}{0.4989968886385938}
			\pgfmathsetmacro{\sxprev}{0.024478175283532327}
			\pgfmathsetmacro{\syprev}{0.058007671902636324}
			\pgfmathsetmacro{\sxyprev}{0.016275673571944073}
			\pgfmathsetmacro{\rprev}{0.4319233997986131}

			\pgfmathsetmacro\betax{\sxyprev / \sxprev}
			\pgfmathsetmacro\alphax{\yprevbar - \betax*\xprevbar}
			\pgfmathsetmacro\Ixux{(1 - \alphax) / \betax}
			\pgfmathsetmacro\Ixlx{(0 - \alphax) / \betax}

			\pgfmathsetmacro\betay{\sxyprev / \syprev}
			\pgfmathsetmacro\alphay{\xprevbar - \betay*\yprevbar}
			\pgfmathsetmacro\Iyuy{(1 - \alphay) / \betay}
			\pgfmathsetmacro\Iyly{(0 - \alphay) / \betay}

			\pgfmathsetmacro\xminx{max(0, (0 - \alphax) / \betax)}
			\pgfmathsetmacro\xmaxx{min(1, (1 - \alphax) / \betax)}

			\pgfmathsetmacro\yminy{max(0, (0 - \alphay) / \betay)}
			\pgfmathsetmacro\ymaxy{min(1, (1 - \alphay) / \betay)}

			\begin{axis}[
				axis equal,
				xlabel={$x$},
				ylabel={$y$},
				axis x line=middle,
				axis y line=middle,
				samples=40,
				xmin=-0.15, xmax=1.15,
				ymin=-0.15, ymax=1.15,
				legend style={anchor=east, at={(1,0.5)}},
				grid,
			]
				\draw[thick, dashed, fill=gray, fill opacity=0.5] (0, 0) rectangle (1, 1);

				\addplot[-,
						domain=\xminx:\xmaxx,
						smooth,
						variable=\x,
						blue
						]
						({\x}, {\betax*\x + \alphax});
				\addlegendentry{\tiny $\mathcal G_{x,n-1}$}
				\addplot[-,
						 domain=\yminy:\ymaxy,
						 smooth,
						 variable=\y,
						 red
						 ]
						 ({\betay*\y + \alphay}, {\y});
				\addlegendentry{\tiny $\mathcal G_{y,n-1}$}

				\addplot[only marks, mark=x] coordinates {
					(0.417022004702574,0.7203244934421581)
					(0.00011437481734488664,0.30233257263183977)
					(0.14675589081711304,0.0923385947687978)
					(0.1862602113776709,0.34556072704304774)
					(0.39676747423066994,0.538816734003357)
					(0.4191945144032948,0.6852195003967595)
					(0.20445224973151743,0.8781174363909454)
					(0.027387593197926163,0.6704675101784022)
					(0.41730480236712697,0.5586898284457517)
					(0.14038693859523377,0.1981014890848788)
				};
				\addlegendentry{\tiny $\left(x_i, y_i \right)$}

				\addplot[only marks, mark=square*] coordinates {
					(\xprevbar,\yprevbar)
				};
				\addlegendentry{\tiny $\left(\bar x_{n-1}, \bar y_{n-1} \right)$}

				\addplot[only marks, mark=*, mark options={fill=blue}] coordinates {(\Ixux,1)};

				\addplot[mark=*] coordinates {(0,1)};
				\node[above left] at (0.0,1.0) {$C_{\ell,u}$};

				\addplot[mark=*] coordinates {(1,1)};
				\node[above right] at (1.00,1.00) {$C_{u,u}$};

				\addplot[mark=*] coordinates {(0,0)};
				\node[below left] at (0.00,0) {$C_{\ell,\ell}$};

				\addplot[mark=*] coordinates {(1,0)};
				\node[below right] at (1,.0) {$C_{u,\ell}$};

				\addplot[mark=*, mark options={fill=blue}] coordinates {(\Ixux,1)};
				\node[left, blue] at (\Ixux,1.05) {$I_{x,u,n-1}$};

				\node[] at (0.5,0.5) {$\mathcal F$};
				\node[] at (0.5,0.95) {$\partial \mathcal U_y$};
				\node[rotate=90] at (0.95,0.25) {$\partial \mathcal U_x$};
				\node[rotate=-90] at (0.05, 0.5) {$\partial \mathcal L_x$};
				\node[] at (0.5, 0.05) {$\partial \mathcal L_y$};

			\end{axis}
		\end{tikzpicture}
		\caption{An illustration of the geometry of $\mathcal F$, $D$, and the intersection and corner points. In this example, only $I_{x,u,n-1} \in \mathcal F$.}
\label{fig:2d}
\end{figure}

\section{Primary sensitivity}
\label{sec:res}

In the following section, we leverage Welford \cite{casella2024statistical} to derive a closed-form function mapping from new points $\left( x_n, y_n \right)$ to the quantity $r \left( D \cup \left\{ \left( x_n, y_n \right) \right\} \right)$. By studying the solutions to this function, we will then be able to trivially compute $\Delta_1 r$ and $\Delta_1 p$.

Substituting Welford's identities into \cref{def:pcc}, simplifying, and then rewriting the correlation as a function $r_n : \mathcal F \to \left[-1, 1 \right]$, we have
\begin{align}
	r_{n}(x_n, y_n) = \frac{n r_{n-1} s_{x, n-1} s_{y, n-1} +  \left( x_n - \bar x_{n-1} \right)  \left( y_n - \bar y_{n-1} \right) }{\sqrt{ \left( n s_{x,n-1}^2 +  \left( x_n - \bar x_{n-1} \right)^2 \right) \left( n s_{y,n-1}^2 +  \left( y_n - \bar y_{n-1} \right)^2 \right)}}.
\end{align}
For brevity, we let $s_x = s_{x, n-1}$, $s_y = s_{y, n-1}$, $\Delta x = x_n - \bar x_{n-1}$, $\Delta y = y_n - \bar y_{n-1}$, and $r = r_{n-1}$, then observe that we have the following derivatives of $r_n$:
\begin{equation}
\begin{cases}
	\frac{\partial r_n}{\partial x_n}  &= \frac{n s_{x} \left(s_{x} \Delta y - r s_{y} \Delta x \right)}{\left(n s_{x}^{2} + \Delta x^{2}\right)^{\frac{3}{2}} \left(n s_{y}^{2} + \Delta y^{2}\right)^{\frac{1}{2}}} \\
	\frac{\partial r_n}{\partial y_n} &= \frac{n s_{y} \left(s_{y} \Delta x - r s_{x} \Delta y \right)}{\left(n s_{x}^{2} + \Delta x^{2}\right)^{\frac{1}{2}} \left(n s_{y}^{2} + \Delta y^{2}\right)^{\frac{3}{2}}} \\
	\frac{\partial^2 r_n}{\partial x_n^2} &= \frac{n s_{x} \left(2 r s_{y} \Delta x^{2} - n r s_{x}^{2} s_{y} - 3 s_{x} \Delta x \Delta y\right)}{\left(n s_{x}^{2} + \Delta x^{2}\right)^{\frac{5}{2}} \left(n s_{y}^{2} + \Delta y^{2}\right)^{\frac{1}{2}}} \\ 
	\frac{\partial^2 r_n}{\partial y_n^2} &= \frac{n s_{y} \left(2 r s_{x} \Delta y^{2} - n r s_{x} s_{y}^{2} - 3 s_{y} \Delta x \Delta y\right)}{ \left(n s_{x}^{2} + \Delta x^{2}\right)^{\frac{1}{2}} \left(n s_{y}^{2} + \Delta y^{2}\right)^{\frac{5}{2}}} \\ 
	\frac{\partial^2 r_n}{\partial x_n \partial y_n} &= \frac{n s_{x} s_{y} \left(n s_{x} s_{y} + r \Delta x \Delta y\right)}{  \left(n s_{x}^{2} + \Delta x^{2}\right)^{\frac{3}{2}} \left(n s_{y}^{2} + \Delta y^{2}\right)^{\frac{3}{2}}}
\end{cases}.
\end{equation}

\begin{proposition}
\label{prop:grad0}
For all $\left( x_n, y_n \right) \in \mathcal F$, $\frac{\partial r_n}{\partial x_n} = 0$ $\left[\frac{\partial r_n}{\partial y_n} = 0\right]$ if and only if $\left( x_n, y_n \right) \in \mathcal G_{x, n-1}$ $\left[ \left( x_n, y_n \right) \in \mathcal G_{y, n-1} \right]$.
\end{proposition}

\begin{proof}
	It suffices to prove our proposition for $x_n$. If $\frac{\partial r_n}{\partial x_n} = 0$, we must have
	\begin{equation}
		s_{x, n-1} \left( y_n - \bar y_{n-1} \right) - r_{n-1} s_{y,n-1} \left( x_n - \bar x_{n-1} \right) = 0,
	\end{equation}
	meaning
	\begin{align*}
		y_n &= r_{n-1} \frac{s_{y, n-1}}{s_{x, n-1}} \left( x_n - \bar x_{n-1} \right) + \bar y_{n-1} \\
			&= \hat \beta_{x, n-1} \cancel{\frac{s_{x, n-1}}{s_{y, n-1}}} \cancel{\frac{s_{y, n-1}}{s_{x, n-1}}}\left( x_n - \bar x_{n-1} \right) + \bar y_{n-1} \\
			&= \hat \beta_{x, n-1} x_n - \hat \beta_{x, n-1} \bar x_{n-1} + \bar y_{n-1} \\
			&=  \hat \beta_{x, n-1} x_n + \hat \alpha_{x, n-1}.
	\end{align*}
	Therefore, $\left( x_n, y_n \right) \in \mathcal G_{x, n-1}$. Our proof of the converse proceeds identically.
\end{proof}

For an illustration of the following, the reader is directed to \cref{fig:3d}. 

\begin{proposition}
\label{prop:boundary}
	\begin{equation}
		\argmax_{\mathcal F} r_n, \ \argmin_{\mathcal F} r_n \in \partial \mathcal F
	\end{equation}
\end{proposition}

\begin{proof}
	Because $r_n$ is continuous, we observe by the extreme value theorem that it must take extremal values along the closed and bounded set $\mathcal F$ \cite{rudin1964principles}. It therefore suffices to prove the demonstrandum for $\argmax_{\mathcal F} r_n$. By way of contradiction, we suppose that $(x_n^*, y_n^*) \in \text{int} \left( \mathcal F \right)$ for all $(x_n^*, y_n^*) = \argmax r_n$.

	By Fermat's theorem \cite{rudin1964principles}, we then know that we must have the first-order condition $\nabla r_n \left( x_n^*, y_n^* \right) = \mathbf 0.$ Furthermore, by \cref{prop:grad0} and the fact that both $\frac{\partial r_n}{\partial x_n} = 0$ and $\frac{\partial r_n}{\partial y_n} = 0$ at $\left( x_n^*, y_n^* \right)$, we know that $\left( x_n^*, y_n^* \right)$ must be in both $\mathcal G_{x, n-1}$ and $\mathcal G_{y, n-1}$. Because both sets are line segments, they must intersect at either an infinite set of points or a single point. We then have two cases:
	\begin{enumerate}
		\item[(i)] $\mathcal G_{x, n-1} = \mathcal G_{y, n-1}$, or
		\item[(ii)] $\mathcal G_{x, n-1} \cap \mathcal G_{y, n-1} = \left\{ \left( \bar x_{n-1}, \bar y_{n-1} \right) \right\}$,
	\end{enumerate}
	where the second case is due to \cref{prop:blueinter}.

	Suppose that (i). Then we have $\nabla r_n \left( x_n, y_n \right) = \mathbf 0$ for all $\left( x_n, y_n \right) \in \mathcal G_{x, n-1} = \mathcal G_{y, n-1}$, in which case, by the fundamental theorem of line integrals \cite{rudin1964principles}, $\mathcal G_{x, n-1} = \mathcal G_{y, n-1}$ must constitute a smooth curve $U \subset \mathcal F$ along which the function $r_n$ is constant.  However, by the intermediate value theorem \cite{rudin1964principles}, we observe that, because $U$ is continuous on $\mathbb R$, it must intersect the boundary at some point $\left( x_n^\prime, y_n^\prime \right) \in \partial \mathcal F$. But because $r_n$ is constant on $U$, we must have $\left( x_n, y_n \right) = \argmax_{\mathcal F} r_n$ for all $\left( x_n, y_n \right) \in U$, including for $(x_n^\prime, y_n^\prime) \in \partial \mathcal F$, which contradicts our assumption that all maxima lie in the interior. \footnote{We note that the comparative complexity of our proof for the first case relative to that for the second is due to the inconclusiveness of the second-derivative test.}

	We must instead have (ii), in which case $\left( x_n^*, y_n^* \right) = \left( \bar x_{n-1}, \bar y_{n-1} \right)$. Because it is a local maximum, second-order optimality conditions must hold \cite{boyd2004convex}, i.e., we must have 
	\begin{equation}
		\text{det} \left( \mathbf H_{r_n} \left(x_n^*, y_n^* \right) \right) = \frac{r_{n-1}^2 -1}{n^2 s_{x,n-1}^2 s_{y,n-1}^2} \geq 0.
	\end{equation}
	But by the contrapositive of \cref{prop:bluesame}, we know that $r_{n-1} \in \left( -1, 1 \right)$, meaning that $\text{det} \left( \mathbf H_{r_n} \left(x_n^*, y_n^* \right) \right) < 0$. Then $\left( x_n^*, y_n^* \right)$ is not a local maximum, which is a contradiction.

	Thus, we must have $\argmax r_n \in \mathcal \partial \mathcal F$, which completes our proof.
\end{proof}

\begin{lemma}
	\label{lem:8point}
	\begin{equation}
		\argmax_{\mathcal F} r_n, \ \argmin_{\mathcal F} r_n \in \mathcal C \cup \mathcal I_{n-1}
	\end{equation}
\end{lemma}

\begin{proof}
	We extend \cref{prop:boundary} by partitioning $\partial \mathcal F$ into its four edges. Without loss of generality, we consider the edge $\partial \mathcal U_y$ and extremum $\left( x_n^*, y_n^* \right)$ and show that $\left( x_n^*, y_n^* \right)  \in \left\{ C_{\ell, u}, C_{u, u}, I_{x, u, n-1} \right\}$. Putting $f \left( x_n \right) := r_n \left( x_n, u_y \right)$ and considering one-dimensional first-order optimality conditions \cite{boyd2004convex}, we know that at least one of the following hold:
	\begin{enumerate}
		\item[(i)] $x_n^* \in \left\{ \ell_x, u_x \right\}$, or
		\item[(ii)] $f^\prime \left( x_n^* \right) = 0$.
	\end{enumerate}
	If (i), then $\left( x_n^*, y_n^* \right) = C_{\ell, u}$ or $\left( x_n^*, y_n^* \right) = C_{u, u}$. If (ii), then by \cref{prop:grad0}, we must have $\left( x_n^*, y_n^* \right) = I_{x, u, n-1}$. This completes our proof.
\end{proof}

Likewise, defining the p-value as a function $p_n: \mathcal F \to [0, 1]$, we obtain similar results.
\begin{proposition}
	\label{prop:station}
	If $\max_{\mathcal F} \left( r_n \right) \geq 0 \geq \min_{\mathcal F} \left( r_n \right)$, then there exists a point $\left(x_n^*, y_n^* \right) \in \mathcal F$ such that $r_n \left(x_n^*, y_n^* \right) = 0$.
\end{proposition}

\begin{proof}
	By \cref{prop:boundary}, we know that $\argmax_{\mathcal F} r_n, \argmin_{\mathcal F} r_n \in \partial \mathcal F$. This allows us to invoke the Poincar\'{e}-Miranda theorem \cite{smart1980fixed} to infer that there exists a stationary point of $r_n$ in $\partial \mathcal F \subset \mathcal F$. \footnote{Technically, Poincar\'{e}-Miranda supposes that the domain of $r_n$ is the unit square as we do in \cref{fig:3d}. However, the PCC is invariant to translation and scaling, and so we are in fact able to suppose that $\mathcal F = \left[0, 1 \right]^2$ without loss of generality. We can then apply \cref{lem:8point} to an arbitrary $\mathcal F$.}
\end{proof}

\begin{proposition}
	\label{prop:pextr}
	\thinspace
	\begin{itemize}
		\item[(I)] If $\max_{\mathcal F} \left( r_n \right) \geq \min_{\mathcal F} \left( r_n \right) \geq 0$, then $\argmax_{\mathcal F} p_n = \argmin_{\mathcal F} r_n$ and $\argmin_{\mathcal F} p_n = \argmax_{\mathcal F} r_n$.
		\item[(II)] If $0 \geq \max_{\mathcal F} \left( r_n \right) \geq \min_{\mathcal F} \left( r_n \right)$, then $\argmax_{\mathcal F} p_n = \argmax_{\mathcal F} r_n$ and $\argmin_{\mathcal F} p_n = \argmin_{\mathcal F} r_n$.
	\end{itemize}
\end{proposition}

\begin{proof}
	It suffices only to prove (I). Due to \cref{cor:mono}, it also suffices only to show that $\argmax_{\mathcal F} |r_n| = \argmax_{\mathcal F} r_n$ and $\argmin_{\mathcal F} |r_n| = \argmin_{\mathcal F} r_n$. We then simply observe that, because both extrema are positive, we must have $|r_n| = r_n$, meaning that the extrema do indeed coincide.
\end{proof}

\begin{lemma}
	\thinspace
	\label{lem:8pointp}
	\begin{itemize}
		\item[(I)] If $\max_{\mathcal F} \left( r_n \right) \geq 0 \geq \min_{\mathcal F} \left( r_n \right)$, then $\max_{\mathcal F} p_n = 1$; otherwise, $\argmax_{\mathcal F} p_n \in \mathcal C \cup \mathcal I_{n-1}$.
		\item[(II)] $\argmin_{\mathcal F} p_n \in \mathcal C \cup \mathcal I_{n-1}$
	\end{itemize}
\end{lemma}

\begin{proof}
	We have three cases:
	\begin{enumerate}
		\item[(i)] $\max_{\mathcal F} \left( r_n \right) \geq 0 \geq \min_{\mathcal F} \left( r_n \right)$,
		\item[(ii)] $\max_{\mathcal F} \left( r_n \right) \geq \min_{\mathcal F} \left( r_n \right) \geq 0$, and
		\item[(iii)] $0 \geq \max_{\mathcal F} \left( r_n \right) \geq \min_{\mathcal F} \left( r_n \right)$.
	\end{enumerate}
	Suppose that (i). By \cref{prop:station}, we will then have a stationary point. The p-value associated with a PCC of 0 is 1, giving us $\max_{\mathcal F} p_n = 1$ and therefore (I). By \cref{cor:mono}, we again know that $\argmin_{\mathcal F} p_n = \argmax_{\mathcal F} |r_n|$, meaning that $\argmin_{\mathcal F} p_n \in \left\{ \argmax_{\mathcal F} r_n, \argmin_{\mathcal F} r_n \right\}$. By \cref{lem:8point}, we conclude that $\argmin_{\mathcal F} p_n \in \mathcal C \cup \mathcal I_{n-1}$, giving us (II).

	If (ii) or (iii), we apply \cref{prop:pextr} and \cref{lem:8point} to obtain both (I) and (II) and complete our proof.
\end{proof}

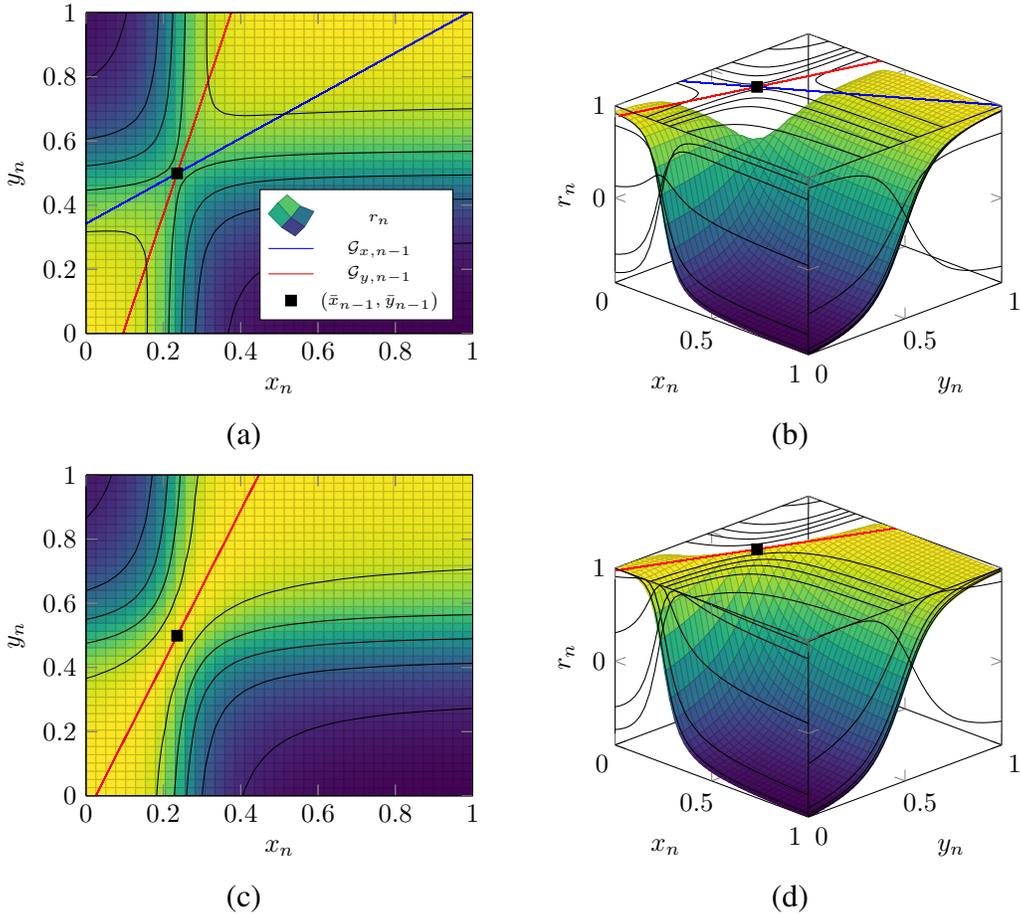
\begin{figure}[h]
    \centering
    \begin{subfigure}[b]{.48\linewidth}
        \centering
		\begin{tikzpicture}
			\pgfmathsetmacro{\n}{10}
			\pgfmathsetmacro{\xprevbar}{0.23556460542404717}
			\pgfmathsetmacro{\yprevbar}{0.4989968886385938}
			\pgfmathsetmacro{\sxprev}{0.024478175283532327}
			\pgfmathsetmacro{\syprev}{0.058007671902636324}
			\pgfmathsetmacro{\sxyprev}{0.016275673571944073}
			\pgfmathsetmacro{\rprev}{0.4319233997986131}

			\pgfmathsetmacro\betax{\sxyprev / \sxprev}
			\pgfmathsetmacro\alphax{\yprevbar - \betax*\xprevbar}
			\pgfmathsetmacro\Ixux{(1 - \alphax) / \betax}
			\pgfmathsetmacro\Ixlx{(0 - \alphax) / \betax}

			\pgfmathsetmacro\betay{\sxyprev / \syprev}
			\pgfmathsetmacro\alphay{\xprevbar - \betay*\yprevbar}
			\pgfmathsetmacro\Iyuy{(1 - \alphay) / \betay}
			\pgfmathsetmacro\Iyly{(0 - \alphay) / \betay}

			\pgfmathsetmacro\xminx{max(0., (0. - \alphax) / \betax)}
			\pgfmathsetmacro\xmaxx{min(1., (1. - \alphax) / \betax)}

			\pgfmathsetmacro\yminy{max(0., (0. - \alphay) / \betay)}
			\pgfmathsetmacro\ymaxy{min(1., (1. - \alphay) / \betay)}

			\begin{axis}[
				view={0}{90},
				axis on top,
				enlargelimits=false,
				xlabel={$x_n$},
				ylabel={$y_n$},
				samples=40,
				domain=0:1,
				domain y=0:1,
				scale=0.75,
				legend style={at={(.95,0.45)}},
				colormap name=viridis,
				3d box=complete,
				enlargelimits=false,
			]
				\addplot3[surf] {
					(\n * \rprev * \sxprev * \syprev + (x - \xprevbar) * (y - \yprevbar)) / sqrt((\n * \sxprev^2 + (x - \xprevbar)^2) * (\n * \syprev^2 + (y - \yprevbar)^2))
				};
				\addlegendentry{\tiny $r_n$}
				\addplot3[-,
						domain=\xminx:\xmaxx,
						smooth,
						variable=\x,
						blue
						]
						({\x}, {\betax*\x + \alphax}, {0});
				\addlegendentry{\tiny $\mathcal G_{x,n-1}$}
				\addplot3[-,
						 domain=\yminy:\ymaxy,
						 smooth,
						 variable=\y,
						 red
						 ]
						 ({\betay*\y + \alphay}, {\y}, {0});
				\addlegendentry{\tiny $\mathcal G_{y,n-1}$}

				\addplot[only marks, mark=square*] coordinates {
					(\xprevbar,\yprevbar)
				};
				\addlegendentry{\tiny $\left(\bar x_{n-1}, \bar y_{n-1} \right)$}
				 \addplot3 [
					contour lua={contour dir=z,
						draw color=black,labels=false},
					z filter/.expression={1},
				] { 
					(\n * \rprev * \sxprev * \syprev + (x - \xprevbar) * (y - \yprevbar)) / sqrt((\n * \sxprev^2 + (x - \xprevbar)^2) * (\n * \syprev^2 + (y - \yprevbar)^2))
				};

			\end{axis}
		\end{tikzpicture}
		\caption{}
    \end{subfigure}%
    \hfill
    \begin{subfigure}[b]{.48\linewidth}
        \centering
        \begin{tikzpicture}
			\pgfmathsetmacro{\n}{10}
			\pgfmathsetmacro{\xprevbar}{0.23556460542404717}
			\pgfmathsetmacro{\yprevbar}{0.4989968886385938}
			\pgfmathsetmacro{\sxprev}{0.024478175283532327}
			\pgfmathsetmacro{\syprev}{0.058007671902636324}
			\pgfmathsetmacro{\sxyprev}{0.016275673571944073}
			\pgfmathsetmacro{\rprev}{0.4319233997986131}

			\pgfmathsetmacro\betax{\sxyprev / \sxprev}
			\pgfmathsetmacro\alphax{\yprevbar - \betax*\xprevbar}
			\pgfmathsetmacro\Ixux{(1 - \alphax) / \betax}
			\pgfmathsetmacro\Ixlx{(0 - \alphax) / \betax}

			\pgfmathsetmacro\betay{\sxyprev / \syprev}
			\pgfmathsetmacro\alphay{\xprevbar - \betay*\yprevbar}
			\pgfmathsetmacro\Iyuy{(1 - \alphay) / \betay}
			\pgfmathsetmacro\Iyly{(0 - \alphay) / \betay}

			\pgfmathsetmacro\xminx{max(-0., (-0. - \alphax) / \betax)}
			\pgfmathsetmacro\xmaxx{min(1., (1. - \alphax) / \betax)}

			\pgfmathsetmacro\yminy{max(0, (0 - \alphay) / \betay)}
			\pgfmathsetmacro\ymaxy{min(1, (1 - \alphay) / \betay)}
			\begin{axis}[
				xlabel={$x_n$},
				ylabel={$y_n$},
				zlabel={$r_n$},
                view={45}{30},
				scale=0.75,
				domain=0:1,
				domain y=0:1,
				samples=40,
				colormap name=viridis,
				3d box=complete,
				enlargelimits=false,
			]
				\addplot3[surf] {
					(\n * \rprev * \sxprev * \syprev + (x - \xprevbar) * (y - \yprevbar)) / sqrt((\n * \sxprev^2 + (x - \xprevbar)^2) * (\n * \syprev^2 + (y - \yprevbar)^2))
				};
				 \addplot3 [
					contour lua={contour dir=y,
						draw color=black,labels=false},
					y filter/.expression={0},
				] { 
					(\n * \rprev * \sxprev * \syprev + (x - \xprevbar) * (y - \yprevbar)) / sqrt((\n * \sxprev^2 + (x - \xprevbar)^2) * (\n * \syprev^2 + (y - \yprevbar)^2))
				};
				 \addplot3 [
					contour lua={contour dir=x,
						draw color=black,labels=false},
					x filter/.expression={1},
				] { 
					(\n * \rprev * \sxprev * \syprev + (x - \xprevbar) * (y - \yprevbar)) / sqrt((\n * \sxprev^2 + (x - \xprevbar)^2) * (\n * \syprev^2 + (y - \yprevbar)^2))
				};
				 \addplot3 [
					contour lua={contour dir=z,
						draw color=black,labels=false},
					z filter/.expression={1},
				] { 
					(\n * \rprev * \sxprev * \syprev + (x - \xprevbar) * (y - \yprevbar)) / sqrt((\n * \sxprev^2 + (x - \xprevbar)^2) * (\n * \syprev^2 + (y - \yprevbar)^2))
				};
				\addplot3[-,
						domain=\xminx:\xmaxx,
						smooth,
						variable=\x,
						blue
						]
						({\x}, {\betax*\x + \alphax}, {1});
				\addplot3[-,
						 domain=\yminy:\ymaxy,
						 smooth,
						 variable=\y,
						 red
						 ]
						 ({\betay*\y + \alphay}, {\y}, {1});
			\addplot3[only marks, mark=square*] coordinates {
				(\xprevbar,\yprevbar,1)
			};
            \end{axis}
        \end{tikzpicture}
		\caption{}
    \end{subfigure}
    \begin{subfigure}[b]{.48\linewidth}
        \centering
		\begin{tikzpicture}
			\pgfmathsetmacro{\n}{10}
			\pgfmathsetmacro{\xprevbar}{0.23556460542404717}
			\pgfmathsetmacro{\yprevbar}{0.4989968886385938}
			\pgfmathsetmacro{\sxprev}{0.024478175283532327}
			\pgfmathsetmacro{\syprev}{0.058007671902636324}
			\pgfmathsetmacro{\sxyprev}{0.016275673571944073}
			\pgfmathsetmacro{\rprev}{1}

			\pgfmathsetmacro\betax{\rprev * \syprev / \sxprev}
			\pgfmathsetmacro\alphax{\yprevbar - \betax*\xprevbar}
			\pgfmathsetmacro\Ixux{(1 - \alphax) / \betax}
			\pgfmathsetmacro\Ixlx{(0 - \alphax) / \betax}

			\pgfmathsetmacro\betay{\rprev * \sxprev / \syprev}
			\pgfmathsetmacro\alphay{\xprevbar - \betay*\yprevbar}
			\pgfmathsetmacro\Iyuy{(1 - \alphay) / \betay}
			\pgfmathsetmacro\Iyly{(0 - \alphay) / \betay}

			\pgfmathsetmacro\xminx{max(0, (0 - \alphax) / \betax)}
			\pgfmathsetmacro\xmaxx{min(1, (1 - \alphax) / \betax)}

			\pgfmathsetmacro\yminy{max(0, (0 - \alphay) / \betay)}
			\pgfmathsetmacro\ymaxy{min(1, (1 - \alphay) / \betay)}
			\begin{axis}[
				view={0}{90},
				axis on top,
				enlargelimits=false,
				xlabel={$x_n$},
				ylabel={$y_n$},
				samples=40,
				domain=0:1,
				domain y=0:1,
				scale=0.75,
				colormap name=viridis,
				3d box=complete,
				enlargelimits=false,
			]

			\addplot3[surf] {
				(\n * \rprev * \sxprev * \syprev + (x - \xprevbar) * (y - \yprevbar)) / sqrt((\n * \sxprev^2 + (x - \xprevbar)^2) * (\n * \syprev^2 + (y - \yprevbar)^2))
			};
			\addplot3[-,
					domain=\xminx:\xmaxx,
					smooth,
					variable=\x,
					blue
					]
					({\x}, {\betax*\x + \alphax}, {0});
			\addplot3[-,
					 domain=\yminy:\ymaxy,
					 smooth,
					 variable=\y,
					 red
					 ]
					 ({\betay*\y + \alphay}, {\y}, {0});

			\addplot[only marks, mark=square*] coordinates {
				(\xprevbar,\yprevbar)
			};
			 \addplot3 [
				contour lua={contour dir=z,
					draw color=black,labels=false},
				z filter/.expression={1},
			] { 
				(\n * \rprev * \sxprev * \syprev + (x - \xprevbar) * (y - \yprevbar)) / sqrt((\n * \sxprev^2 + (x - \xprevbar)^2) * (\n * \syprev^2 + (y - \yprevbar)^2))
			};

			\end{axis}
		\end{tikzpicture}
		\caption{}
    \end{subfigure}%
    \hfill
    \begin{subfigure}[b]{.48\linewidth}
        \centering
        \begin{tikzpicture}
			\pgfmathsetmacro{\n}{10}
			\pgfmathsetmacro{\xprevbar}{0.23556460542404717}
			\pgfmathsetmacro{\yprevbar}{0.4989968886385938}
			\pgfmathsetmacro{\sxprev}{0.024478175283532327}
			\pgfmathsetmacro{\syprev}{0.058007671902636324}
			\pgfmathsetmacro{\sxyprev}{0.016275673571944073}
			\pgfmathsetmacro{\rprev}{1}

			\pgfmathsetmacro\betax{\rprev * \syprev / \sxprev}
			\pgfmathsetmacro\alphax{\yprevbar - \betax*\xprevbar}
			\pgfmathsetmacro\Ixux{(1 - \alphax) / \betax}
			\pgfmathsetmacro\Ixlx{(0 - \alphax) / \betax}

			\pgfmathsetmacro\betay{\rprev * \sxprev / \syprev}
			\pgfmathsetmacro\alphay{\xprevbar - \betay*\yprevbar}
			\pgfmathsetmacro\Iyuy{(1 - \alphay) / \betay}
			\pgfmathsetmacro\Iyly{(0 - \alphay) / \betay}

			\pgfmathsetmacro\xminx{max(0, (0 - \alphax) / \betax)}
			\pgfmathsetmacro\xmaxx{min(1, (1 - \alphax) / \betax)}

			\pgfmathsetmacro\yminy{max(0, (0 - \alphay) / \betay)}
			\pgfmathsetmacro\ymaxy{min(1, (1 - \alphay) / \betay)}
            \begin{axis}[
                xlabel={$x_n$},
                ylabel={$y_n$},
				zlabel={$r_n$},
                view={45}{30},
                scale=0.75,
                domain=0:1,
                domain y=0:1,
                samples=40,
				colormap name=viridis,
				3d box=complete,
				enlargelimits=false,
            ]
			\addplot3[surf] {
				(\n * \rprev * \sxprev * \syprev + (x - \xprevbar) * (y - \yprevbar)) / sqrt((\n * \sxprev^2 + (x - \xprevbar)^2) * (\n * \syprev^2 + (y - \yprevbar)^2))
			};
			 \addplot3 [
				contour lua={contour dir=y,
					draw color=black,labels=false},
				y filter/.expression={0},
			] { 
				(\n * \rprev * \sxprev * \syprev + (x - \xprevbar) * (y - \yprevbar)) / sqrt((\n * \sxprev^2 + (x - \xprevbar)^2) * (\n * \syprev^2 + (y - \yprevbar)^2))
			};
			 \addplot3 [
				contour lua={contour dir=x,
					draw color=black,labels=false},
				x filter/.expression={1},
			] { 
				(\n * \rprev * \sxprev * \syprev + (x - \xprevbar) * (y - \yprevbar)) / sqrt((\n * \sxprev^2 + (x - \xprevbar)^2) * (\n * \syprev^2 + (y - \yprevbar)^2))
			};
			 \addplot3 [
				contour lua={contour dir=z,
					draw color=black,labels=false},
				z filter/.expression={1},
			] { 
				(\n * \rprev * \sxprev * \syprev + (x - \xprevbar) * (y - \yprevbar)) / sqrt((\n * \sxprev^2 + (x - \xprevbar)^2) * (\n * \syprev^2 + (y - \yprevbar)^2))
			};
			\addplot3[-,
					domain=\xminx:\xmaxx,
					smooth,
					variable=\x,
					blue
					]
					({\x}, {\betax*\x + \alphax}, {1});
			\addplot3[-,
					 domain=\yminy:\ymaxy,
					 smooth,
					 variable=\y,
					 red
					 ]
					 ({\betay*\y + \alphay}, {\y}, {1});
			\addplot3[only marks, mark=square*] coordinates {
				(\xprevbar,\yprevbar,1)
			};
            \end{axis}
        \end{tikzpicture}
		\caption{}
    \end{subfigure}
	\caption{(a) The same dataset in \cref{fig:2d} with both BLUE lines superimposed on the function $r_n$. (b) A 3D view of $r_n$ for the same dataset in \cref{fig:2d}. Critically, $\left( \bar x_{n-1}, \bar y_{n-1} \right)$ is a saddle point. As indicated by the projected contour lines of $r_n$, we have stationary points, the existence of which are formally provable via Poincar\'{e}-Miranda \cite{smart1980fixed}. (c) A dataset similar to that in \cref{fig:2d}, albeit such that $r \left( D \right) = 1$. Both BLUE lines coincide. (d) A 3D view of the same dataset in \cref{fig:2d}(c). The BLUE lines form a smooth curve $U \subset \mathcal F$ along which $\nabla r_n = \mathbf 0$ and therefore $r_n$ is constant.}
	\label{fig:3d}
\end{figure}

\section{Main algorithm}
\label{sec:algo}

We are now able to develop our key algorithm, \textproc{ComputePrimarySensitivities} (\cref{alg}). For its implementation, we assume that $|D| = n$, not $n-1$ as we did above. For the pseudocode below, we write $\textproc{Welford}$ for the online summation of Welford's parameters and \textproc{StudentsT} to perform the statistical test in \cref{thm:ttest}. We use our earlier recurrence relation for $r_{n+1}$ to update $D$ merely for brevity; the update rules in \cref{thm:welford} could be employed just as appropriately.

Leveraging \cref{lem:8point}, it computes each intersection point, determines which of them are feasible, and then adds these along with the corner points to $D$ to determine which induces the largest changes to the PCC, thereby computing $\Delta_1 r$. Afterwards, we test for whether $\max r_{n+1} \geq 0 \geq \min r_{n+1}$ in order to determine whether we have $\max p_{n+1} = 1$ by \cref{lem:8pointp}. If so, then either $\Delta_1 p = 1 - p_n$ or it is directly derivable from the corner and intersection points, again by the same lemma.

\begin{theorem}
	\textproc{ComputePrimarySensitivities} (\cref{alg}) correctly returns $\Delta_1 r$ and $\Delta_1 p$.
\end{theorem}

\begin{proof}
	This follows trivially from \cref{lem:8point} and \cref{lem:8pointp}.
\end{proof}

\begin{theorem}
	\label{thm:linear}
	\textproc{ComputePrimarySensitivities} (\cref{alg}) runs in $O(n)$ time.
\end{theorem}

\begin{proof}
	Each of lines 3-9, 11-14, 16-17, and 20-21 will require $O(1)$ time. The loops and set comprehensions in lines 10, 18, and 19 will run for a constant number of iterations given that the number of corner and intersection points is bounded. Therefore, the algorithm is dominated by the call to \textproc{Welford} in line 2, which will require a single $O(n)$ pass to sum all relevant quantities \cite{welford1962note}.
\end{proof}

\begin{corollary}
	Suppose that the Welford parameters of $D$ are computed in advance. Then $\Delta_1 r$ and $\Delta_1 p$ can be computed in $O(1)$ time.
\end{corollary}

\begin{proof}
	We repeat our proof for \cref{thm:linear}, albeit excluding \textproc{Welford} in line 2, resulting in a $O(1)$ runtime.
\end{proof}

\begin{algorithm}
\caption{Calculates primary sensitivities of $D$ within $\mathcal F$.}
\label{alg}
\begin{algorithmic}[1]
\Procedure{ComputePrimarySensitivities}{$D, \mathcal F$}
	\State $n, \ \bar x_n, \ \bar y_n, \ s_{x, n}, \ s_{y, n}, \ s_{xy, n} \gets \Call{Welford}{D}$
	\State $\hat{\beta}_{x,n}, \ \hat{\beta}_{y, n} \gets \frac{s_{xy, n}}{s_{x, n}}, \ \frac{s_{xy, n}}{s_{y, n}}$
	\State $\hat{\alpha}_{x,n}, \ \hat{\alpha}_{y, n} \gets \bar y_n - \hat{\beta}_{x,n}, \ \bar x_n - \hat{\beta}_{y, n} \bar y_n$
    \State $\mathcal C \leftarrow \left\{(x, y) \mid (x = \ell_x \text{ or } x = u_x) \text{ and } (y = \ell_y \text{ or } y = u_y)\right\}$
	\State $\mathcal I_n \leftarrow \left\{ \left(\hat{\beta}_{x, n}^{-1} \left(y - \hat{\alpha}_{x, n} \right), y \right) \mid y = \ell_y \text{ and } y = u_y \right\}$
	\State $\mathcal I_n \leftarrow \mathcal I_n \cup \left\{ \left(x, \hat{\beta}_{y, n}^{-1} \left(x - \hat{\alpha}_{y, n} \right) \right) \mid x = \ell_x \text{ and } x = u_x \right\}$
	\State $\mathcal I_n \gets \left\{ (x, y) \in \mathcal I_n: \ell_x \leq x \leq u_x \text{ and } \ell_y \leq y \leq u_y \right\}$
	\State $\mathbf r_{n+1}, \ \mathbf p_{n+1} \leftarrow \emptyset, \ \emptyset$
	\ForEach {$(x_{n+1}, y_{n+1}) \in \mathcal C \cup \mathcal I_n$}
		\State $\Delta x, \ \Delta y \gets x_{n+1} - \bar x_n, \ y_{n+1} - \bar y_n$
		\State $r_{n+1} \gets \frac{(n+1) s_{x,n} s_{y,n} + \Delta x \Delta y}{\sqrt{\left( (n+1) s_{x,n}^2 + \Delta x^2 \right) \left( (n+1) s_{y,n}^2 + \Delta y^2 \right)}}$
		\State $\mathbf r_{n+1} \leftarrow \mathbf r_{n+1} \cup \left\{ r_{n+1} \right\}$
		\State $\mathbf p_{n+1} \leftarrow \mathbf p_{n+1} \cup \left\{ \Call{StudentsT}{r_{n+1}, n} \right\}$
    \EndFor
	\State $r_n \gets \frac{s_{xy, n}}{s_{x, n} s_{y, n}}$
	\State $p_n \gets \Call{StudentsT}{r_n, n}$
	\State $\Delta_1 r \gets \max \left( \left\{ \left| r_n - r_{n+1} \right|: r_{n+1} \in \mathbf r_{n+1}  \right\} \right)$
	\State $\Delta_1 p \gets \max \left( \left\{ \left| p_n - p_{n+1} \right|: p_{n+1} \in \mathbf p_{n+1}  \right\} \right)$
	\If {$\max(\mathbf r_{n+1}) \geq 0 \geq \min(\mathbf r_{n+1})$}
		\State $\Delta_1 p \gets \max \left( \Delta_1 p, 1 - p_n \right)$
	\EndIf
	\State \Return $\Delta_1 r, \ \Delta_1 p$
\EndProcedure
\end{algorithmic}
\end{algorithm}

\section{Experiments}
\label{sec:exp}

\subsection{Synthetic data}

\begin{figure}[h]
	\centering
	\includegraphics[width=\textwidth]{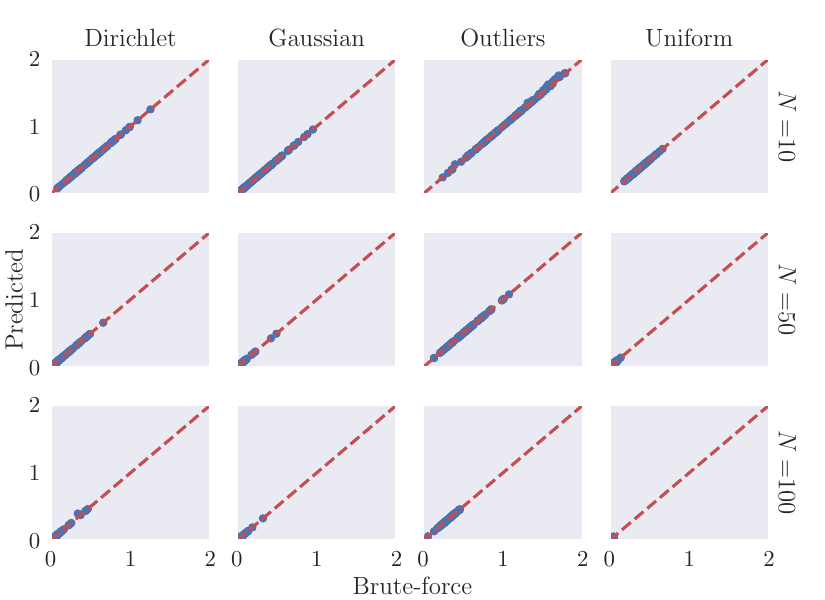}
	\caption{$\Delta_1 r$ for various distributions as computed by both \textproc{ComputePrimarySensitivities} (\cref{alg}) and brute-force grid search.}
	\label{fig:res}
\end{figure}

In order to empirically test the correctness of our algorithm, we generated synthetic datasets of varying sizes and compared the results of \textproc{ComputePrimarySensitivities} (\cref{alg}) to brute-force grid search for the four following distributions. First, we sampled each coordinate from an independent uniform distribution on $[-10, 10]$. Second, we drew from bivariate Gaussian distributions with means of $\boldsymbol \mu = \mathbf 0$ and covariance matrices given by $\boldsymbol \Sigma = \mathbf A^\intercal \mathbf A$, where each element in $\mathbf A \in \left[0, 1 \right]^{2 \times 2}$ was randomly generated from a uniform distribution at each iteration (multiplication by its own transpose was intended to enforce symmetry and positive definiteness). Third, we sampled from Dirichlet distributions with each component $\alpha_i$ of the parameter vector $\boldsymbol \alpha \in \mathbb R^3$ generated by sampling a uniform distribution on $\left[0, 10 \right]$. Fourth, we again generated bivariate Gaussian distributions as above, but substituted 10\% of each sample with outliers generated from a uniform distribution on $\left[-30, 30 \right]$ \cite{virtanen2020scipy}.

For each distribution, we sampled $N \in \left\{ 10, 50, 100 \right\}$ points, and for each distribution and $N$, repeated 100 separate experiments. The feasible region was simply set to the bounding box of each dataset, meaning that its size varied dynamically with each experiment. For grid search, we sampled 10 equidistant points along each axis of the feasible region for a total of 100 candidates; larger datasets and finer grids were not employed due to hardware limitations. Likewise, due to the necessity of high numerical precision for computing small changes in p-value, we opted to analyze only $\Delta r_1$ for these broad initial tests. All experiments were conducted in Numpy \cite{Harris2020} and SciPy \cite{virtanen2020scipy}; \textproc{ComputePrimarySensitivities} (\cref{alg}) itself was implemented primarily in C \cite{kernighan2002c} for efficiency with a Cython \cite{behnel2010cython} front-end for ease of use. 

Results (\cref{fig:res}) indicate that \textproc{ComputePrimarySensitivities} (\cref{alg}) consistently predicted values extremely close to that of grid search. For 98.28\% of the synthetic datasets, the sensitivities were nearly identical $\bigl($i.e., within a relative tolerance of $1 \times 10^{-5} \bigr)$. We conclude that our method correctly predicts the sensitivities of every dataset, with the occasional discrepancies between the two predictions reflecting floating point error or the coarseness of the grid.

\subsection{Real-world example: The Great Recession}

Among the many use cases for efficiently computing the sensitivity of the PCC, financial markets offer an excellent application for investors weighing portfolio decisions based on relationships between various assets \cite{campbell2008predicting, ang2002asymmetric, fama1992cross}. During the Great Recession of 2008, the financial sector became the epicenter of extreme market volatility \cite{reinhart2009aftermath}, prompting investors to closely monitor rapidly evolving market dynamics \cite{adrian2010liquidity}, consider rotating their portfolios to alternative sectors, and otherwise mitigate their exposure to unstable economic conditions \cite{ang2002asymmetric}. A key factor in portfolio stress-testing was the correlation between broad market indices on one hand and hedge funds, insurance companies, and other financial services on the other. The rationale was that in periods of economic crisis emanating from the financial sector, these services would exhibit far greater instability than the broader markets, suggesting that diverse portfolios would be more likely to mitigate financial risk \cite{reinhart2009aftermath, adrian2010liquidity}. However, the correlation between the financial sector and other markets was highly unstable on even a daily basis in 2008 \cite{reinhart2009aftermath}, meaning that primary sensitivities of correlations between financial instruments could have been highly informative.

For example, we consider the primary sensitivities of the SPDR S\&P 500 ETF Trust (SPY) \cite{SPY2008} and Financial Select Sector SPDR Fund (XLF) \cite{XLF2008} in September, 2008, when the DOW Jones Industrial Average plummeted to historic lows on the 29th \cite{reinhart2009aftermath}. Computing the PCC between the closing values of SPY and XLF during the week leading up to that date \cite{yfinance, virtanen2020scipy}, we find that the sample PCC of our dataset is $r_n(D) = 0.58028$, corresponding to a p-value of $p_n(D) = 0.30502$ (\cref{tab:res}).

While we might be inclined to have very low confidence in a linear relationship between the financial sector and the broader market, we suspect that this could change within a single day. In order to prepare for this possibility, we compute the primary sensitivities of the correlation between SPY and XLF. Although our data do not have a clear feasible region, we set lower bounds of 0 for both instruments and, for their upper bounds, posit that neither will exceed its respective maximum encountered in the foregoing week. Using our framework, we find that the primary sensitivities of $D$ are $\Delta_1 r(D, \mathcal F) = 1.15630$ and $\Delta_1 p(D, \mathcal F) = 0.69498$. In other words, our analyses indicate that by the end of trading hours on the 30th, it is fully possible that radical changes in market conditions may suddenly nullify the trend of non-linearity between SPY and XLF, rendering the PCC significant and prompting drastic changes in investment strategy. Surely enough, both indices plummet on the 29th, bringing the p-value to $4.91 \times 10^{-3}$, well within our predicted primary sensitivity.

To confirm the output of our algorithm, we compared the results to that of brute-force search across a much finer grid of 10,000 total points to increase precision for the p-value calculations. We find that our method predicts the correct optima, thereby validating our theory on empirical data. We also observed that the sensitivity $\Delta_1 p = 1 - p_n$ as per \cref{lem:8pointp}; the slight inaccuracy of the brute-force estimate is therefore likely due to the coarseness of the grid.

\begin{table}[ht]
\centering
\caption{Predicted vs. Brute-Force Sensitivity Estimates}
\label{tab:sensitivities}
\begin{tabular}{@{}lccc@{}}
\toprule
Metric & Value & Predicted $\Delta_1$ & Brute-Force $\Delta_1$ \\ \midrule
$r_n$ & 0.58028 & 1.15630 & 1.15630 \\ \midrule
$p_n$ & 0.30502 & 0.69498 & 0.69491 \\ \bottomrule
\end{tabular}
\label{tab:res}
\end{table}

\section{Discussion}
\label{sec:disc}

Practical applications of our approach are vast, extending to many fields in which dynamic datasets require reliable collinearity measures. In bioinformatics, the strength of associations between genetic markers and  phenotypes in genome-wide association studies (GWAS) may be subject to change if data are mined from expanding biobanks \cite{tam2019benefits, wang2005genome, littlejohns2020uk}. In climatology, associations between temperature and various environmental outcomes like ground water availability are famously subject to reevaluation in light of ongoing studies \cite{held2006robust, taylor2013ground, morice2012quantifying}. In behavioral science and other experimental disciplines, similar analyses may be employed to halt experiments prematurely if the strength of association between variables is determined to be sufficiently robust to new subject data \cite{crocker1986introduction, montgomery2017design, raykov2011introduction}. In the technology sector, dynamic environments like streaming platforms draw from real-time analytics to understand potentially unstable relationships between data streams, user preferences, internet traffic, etc. \cite{nguyen2008survey, isinkaye2015recommendation} In differential privacy, the local sensitivity of metrics like the PCC is key for methods like Propose-Test-Release (PTR) \cite{dwork2014algorithmic}. 

As to the theoretical doors opened by our work, future directions are also numerous. First, in many optimization tasks like linear programming, feasible regions are often presumed to be generally polygonal (or, rather, polyhedral for higher-dimensional tasks) as opposed to simply rectangular \cite{dantzig2002linear}. Arguably, the simplicity of our assumed feasible region is one of the primary limitations of our results. However, we conjecture that propositions similar to those proven above might be straightforwardly adapted for arbitrary piecewise linear boundaries, with basic feasible solutions playing roles similar to corner points in the case of a rectangular domain.

We also speculate that $k$-ary sensitivity can be efficiently computed as well by extending our theory further. In analyses not explicated in this paper, we were able to prove that the worst-case sequence of $k$ points to add to a bounded dataset must conform to certain structures, e.g., that each point in the sequence must also be a corner or intersection point. These patterns may offer an opportunity to compute the worst-case sequence via a combinatorial optimization approach.

Finally, our method of leveraging Welford's identities may be extended to other linear metrics such as covariance and the parameters of the best linear unbiased estimators (BLUE) for simple linear regression \cite{casella2024statistical}. For higher-dimensional datasets, similar analyses might be even extended to covariance matrices or principal components \cite{casella2024statistical}.

\section{Conclusion}
\label{sec:conc}

In this work, we develop a rigorous theoretical framework for analyzing changes to the Pearson correlation coefficient (PCC) \cite{pearson1895vii} and its p-value induced by the addition of new data points, as well an algorithm for computing these changes in practice. Our key contributions consist of formulating a recurrence relation for the PCC and p-value via identities vital to Welford's online algorithm \cite{welford1962note}, deriving closed-form solutions to this relation, and finally demonstrating that these solutions can be straightforwardly computed in linear or even constant time, depending on whether key parameters are precomputed. In addition to formal proofs of correctness, we validate our method on synthetic and real-world data to empirically confirm its efficacy.

Our code is made available at \href{https://github.com/marc-harary/sensitivity}{\texttt{https://github.com/marc-harary/sensitivity}} for both practical usage and ongoing algorithmic development in this area.

\bibliographystyle{plain}
\bibliography{refs}

\end{document}